\documentclass[a4paper, 9pt, leqno]{amsart}

\usepackage{multicol, caption}
\usepackage{defs}

\newenvironment{Figure}
{\par\medskip\noindent\minipage{\linewidth}}
{\endminipage\par\medskip}

\usepackage{url}

\title{Characterization of Maximum Hands-Off Control}
\author{Debasish Chatterjee}
\author{Masaaki Nagahara}
\author{Daniel Quevedo}
\author{K. S. Mallikarjuna Rao}

\begin{document}

\begin{abstract}
Maximum hands-off control aims to maximize the length of time over which zero actuator values are applied to a system when executing specified control tasks. To tackle such problems, recent literature has investigated optimal control problems which penalize the size of the support of the control function and thereby lead to desired sparsity properties.  This article gives the exact set of necessary conditions for a maximum hands-off optimal control problem using an \(\Lp 0\)-(semi)norm, and also provides sufficient conditions for the optimality of such controls. Numerical example illustrates that adopting an \(\Lp 0\) cost leads to a sparse control, whereas an \(\Lp  1\)-relaxation in singular problems leads to a non-sparse solution.%
\end{abstract}

\maketitle

	\markboth{Maximum Hands-off control: Existence and characterization}{D.\ Chatterjee, M.\ Nagahara, D.\ Quevedo, and K.\ S.\ Mallikarjuna Rao}



\begin{multicols}{2}


\newcommand{\hfs}{\hfill\ensuremath{\square}} 

\section{Introduction}
\label{s:intro}

	Motivated by a diverse array of applications in automotive industry, railway vehicles, and networked control, the recent works \cite{nagque13a,ref:NagQueNes-14} dealt in detail with the concept of \emph{maximum hands-off control}. The purpose of maximum hands-off control is to design actuator signals which are most often zero, but nonetheless achieve given control objectives. This motivates the use of instantaneous cost functions where the control effort is penalized via the \(\Lp 0\)-(semi)norm, thereby leading to a \emph{sparse} control function, cf.\ \cite{hargal13a,nagque14a,schebe11a,bhabac11a, ref:CapForPicTre-13, ref:CapForPicTre-15, ref:LinFarJov-13}. Sparse controls are of great importance in situations where a central processor must be shared by different controllers, and sparse control is a new and emerging area of research, including applications in the theory of control of partial differential equations \cite{ref:CasTro-14, ref:CasHerWac-12, ref:Sta-09, ref:VosMau-06, ref:ClaItoKun-15}.

	Due to the discontinuous and non-convex nature of the instantaneous cost function in \(\Lp 0\)-optimal control problems, solving such problems is in general difficult. Hence, the precursor article \cite{ref:NagQueNes-14} focused  on relaxations to the problem, akin to methods used in compressed sensing applications \cite{don06a}. To be more precise, \cite{ref:NagQueNes-14} examined smooth and convex relaxations of the maximum hands-off control problem, including considering an \(\Lp 1\)-cost and regularizations with an \(\Lp 2\)-cost to obtain smooth hands-off control. (It is a well-known and classical result that under ``nonsingularity'' assumptions on the control system \cite[Chapter 8]{ref:AthansFalb-66}, \(\Lp 1\)-costs lead to sparse solutions in the control. However, in singular problem instances, it is unclear whether \(\Lp 1\)-regularizations lead to sparse solutions.) The exact \(\Lp 0\)-optimal control problem was not investigated in \cite{ref:NagQueNes-14}.

	The purpose of the present article is to complement \cite{ref:NagQueNes-14} by directly dealing with the underlying non-smooth and non-convex \(\Lp 0\)-optimal control problem without the aid of smooth or convex relaxations. We will focus on nonlinear controlled dynamical systems of the form
	\begin{equation}
	\label{e:sys}
		\dot z(t) = \phi\bigl(z(t), u(t)\bigr)
	\end{equation}
	with state $z$, input $u$ and where  \(\phi:\R^d\times\R^m\lra\R^d\) is a continuously differentiable map describing the open-loop system dynamics.  The maximum hands-off control problem aims to minimize the support of the control map, or in other words, maximize the time duration over which the control map is exactly zero.

	In other words, given real numbers \(\tinit, \tfin\in\R\) with \(\tinit < \tfin\), vectors \(A, B\in\R^d\), a compact set \(\admact\subset\R^m\) containing \(0\in\R^m\) in its interior, we consider the optimal control problem
	\begin{equation}
	\label{e:opt problem}
	\begin{aligned}
		\minimize_{u}	& \quad \norm{u}_{\Lp 0([\tinit, \tfin])}\\
		\sbjto	 & \quad \begin{cases}
					\dot z(t) = \phi\bigl(z(t), u(t)\bigr)	\text{ for a.e.\ }t\in[\tinit, \tfin],\\
					z(\tinit) = A,\quad z(\tfin) = B,\\
					u:[\tinit, \tfin]\lra\admact \text{ Lebesgue measurable}.
			\end{cases}
	\end{aligned}
	\end{equation}
	Here the \(\Lp 0\)-(semi)norm\footnote{Note that our choice of calling the map \(u\mapsto \norm{u}_{\Lp 0([\tinit, \tfin])}\) a seminorm is technically not precise because this map does not satisfy the positive homogeneity property despite being positive definite and satisfying the triangle inequality. However, here we choose to overload the term seminorm in favour of being pedantic.} of a map \(u:[\tinit, \tfin]\lra\admact\) is defined by the Lebesgue measure of the support of \(u\), i.e.,
	\[
		\norm{u}_{\Lp 0([\tinit, \tfin])} \Let \Leb\Bigl(\bigl\{s\in[\tinit, \tfin]\,\big|\, u(s) \neq 0\bigr\}\Bigr).
	\]
	Observe that if the minimum time to transfer the system states from \(z(\tinit) = A\) to \(z(\tfin) = B\) is larger than the given duration \(\tfin - \tinit\), then the optimal control problem \eqref{e:opt problem} has no solution. Thus, a standing assumption used throughout this work is that there is a feasible solution to \eqref{e:opt problem}. In other words, despite the limited control authority described by the compact set \(\admact\), we shall assume that it is possible to steer the system states from \(A\) to \(B\) in finite time \(\tfin - \tinit\). Observe also that, unlike minimum attention control \`a la \cite{brocke97}, the optimal control problem \eqref{e:opt problem} does not penalize the rate of change of the control. Nonetheless, \eqref{e:opt problem} can be viewed through the looking glass of least attention in the sense that the control is `active' for the least duration of time. The current work investigates optimality in \eqref{e:opt problem} using a nonsmooth maximum principle as summarized in \cite[Chapter 22]{ref:Cla-13}.

	The main contributions and outline of this article are given below:
	\begin{enumerate}[label=(\roman*), leftmargin=*, widest=iii, align=right]
        \item  We show that \eqref{e:opt problem} can be recast in the form of an optimal control problem involving an integral cost with a discontinuous cost function. \label{contrib:1} We apply a non-smooth Pontryagin maximum principle directly to problem \eqref{e:opt problem} and obtain an exact set of necessary conditions for optimality. This result is presented in \secref{s:exact solution}. It characterizes solutions to \eqref{e:opt problem} provided that they exist.
		\item \secref{s:example} sheds further insight into the case where the system dynamics in \eqref{e:sys} are linear. This section also illustrates that, perhaps contrary to intuition, in singular problem instances, \(\Lp 1\)-relaxations may fail to give sparse controls; cf.\ \cite[Chapter 8]{ref:AthansFalb-66}.
		\item The Pontryagin maximum principle gives necessary conditions for an extremum. Naturally, any state-action trajectory satisfying the Pontryagin maximum principle is not necessarily optimal. In \secref{s:exact solution} we provide conditions under which the necessary conditions are also sufficient for optimality. 
		Our proof of optimality follows from inductive methods in optimal control.
	\end{enumerate}  

\paragraph*{Notation:} The notations employed in this article are standard. The Euclidean norm of a vector $z$, belonging to the \(d\)-dimensional Euclidean space \(\R^d\), is denoted by \(\norm{z}\); vectors are treated as column vectors. For a set \(S\) we let \(z\mapsto \indic{S}(z)\) denote the indicator (characteristic) function of the set \(S\) defined to be \(1\) if \(z\in S\) and \(0\) otherwise.

\begin{remark}
	\label{r:comparison}
		The version of the maximum hands-off control problem posed in
                \cite{nagque13a,ref:NagQueNes-14} is slightly different from the one we
                examine in \eqref{e:opt problem} above. Indeed,  
                \cite{ref:NagQueNes-14} studies the following problem: 
		\begin{equation}
		\label{e:opt problem in NagQueNes-14}
		\begin{aligned}
			\minimize_{u}	& \quad \frac{1}{b-a}\sum_{i=1}^m \lambda_i \norm{u_i}_{\Lp 0([\tinit, \tfin])}\\
			\sbjto	& \quad \begin{cases}
						\dot z(t) = \phi\bigl(z(t), u(t)\bigr)	\text{ for a.e.\ }t\in[\tinit, \tfin],\\
						z(\tinit) = A,\quad z(\tfin) = B,\\
						u:[\tinit, \tfin]\lra\admact\text{ Lebesgue measurable},
				\end{cases}
		\end{aligned}
		\end{equation}
		where \(\{\lambda_i\}_{i=1}^m\) are given positive weights. This cost function   features the controls of a multivariable plant as additive terms. In contrast, and by noting that
		\[
			\int_{\tinit}^{\tfin} \indic{\{0\}}(u(s))\,\dd s = \int_{\tinit}^{\tfin} \prod_{i=1}^m \indic{\{0\}}(u_i(s))\,\dd s,
		\]
		(where the \(0\) on the left-hand side belongs to \(\R^m\) and the one on the right-hand side belongs to \(\R\),) the cost function \eqref{e:opt problem} features a multiplicative form  in the controls. The techniques exposed for \eqref{e:opt problem} in the sequel carry over in a straightforward fashion to \eqref{e:opt problem in NagQueNes-14}. In order not to blur the message of this article, we stick to the simpler case of \eqref{e:opt problem}. \hfs
	\end{remark}

\section{Necessary Conditions for Optimality}
\label{s:exact solution}

	By definition, we have
	\begin{equation}
	\label{e:lzeronorm as cost}
		\norm{u}_{\Lp 0([\tinit, \tfin])} = \tfin - \tinit - \int_{\tinit}^{\tfin} \indic{\{0\}}(u(s))\,\dd s.
	\end{equation}
	Since \(\tinit\) and \(\tfin\) are fixed, the minimization of
        \(\norm{u}_{\Lp 0([\tinit, \tfin])}\) in \eqref{e:opt problem} is
        equivalent to the minimization of \(-\int_{\tinit}^{\tfin}
        \indic{\{0\}}(u(s))\,\dd s\). In view of this, we rewrite the optimal control problem \eqref{e:opt problem} as
	\begin{equation}
	\label{e:Clarke opt problem}
	\begin{aligned}
		\minimize_{u}	& \quad -\int_{\tinit}^{\tfin} \indic{\{0\}}(u(s))\,\dd s\\
		\sbjto	& \quad \begin{cases}
				\dot z(t) = \phi\bigl(z(t), u(t)\bigr)\quad \text{for a.e.\ }t\in[\tinit, \tfin],\\
				z(\tinit) = A, \quad z(\tfin) = B,\\
				u:[\tinit, \tfin]\lra\admact \text{ Lebesgue measurable}.
			\end{cases}
	\end{aligned}
	\end{equation}
	We have the following Proposition:
	\begin{proposition}
	\label{p:exact solution}
		Associated to every solution \([\tinit, \tfin]\ni t\mapsto \bigl(z\opt(t), u\opt(t)\bigr)\) to \eqref{e:opt problem}  there exist an absolutely continuous curve \([\tinit, \tfin]\ni t\mapsto p(t)\in\R^d\) and a number \(\eta = 0\) or \(1\) such that for a.e.\ \(t\in[\tinit, \tfin]\):
			\begin{equation}
			\label{e:exact:solution}
				\left\{
				\begin{aligned}
					\dot z\opt(t)& = \phi\bigl(z\opt(t), u\opt(t)\bigr),\quad z\opt(\tinit) = A,\; z\opt(\tfin) = B,\\
					\dot p(t)	& = -\Bigl( \partial_z \phi\bigl(z\opt(t), u\opt(t)\bigr) \Bigr)\transp p(t),\\
					u\opt(t)& \in \argmax_{v\in \admact}\Bigl\{\inprod{p(t)}{\phi\bigl(z\opt(t), v\bigr)} + \eta \indic{\{0\}}(v)\Bigr\},
				\end{aligned}
				\right.\\
			\end{equation}
			and 
			\begin{equation}
			\label{e:exact:solution 2}
				\bigl(\eta, p(t)\bigr)\neq (0, 0)\in\R\times\R^d\quad \text{for all }t\in[\tinit, \tfin].
			\end{equation}
	\end{proposition}

	A proof of Proposition \ref{p:exact solution} is provided in Appendix \ref{s:app}.

	\begin{remark}
	\label{r:solutions}
		Proposition \ref{p:exact solution} gives a set of necessary conditions for optimality of state-action trajectories \(t\mapsto\bigl(z\opt(t), u\opt(t)\bigr)\) in the same spirit as the standard first order necessary conditions for an optimum in a finite-dimensional optimization problem. We see that the ordinary differential equations (o.d.e.'s) describing the system state $z\opt$ and its adjoint $p$ constitute a set of \(2d\)-dimensional o.d.e.'s with \(2d\) constraints. This amounts to a well-defined boundary value problem in the sense of Carath\'eodory \cite[Chapter 1]{ref:Fil-88}. Indeed, the control map \(u\opt\) is Lebesgue measurable, and depends parametrically on \(p\); therefore, the right-hand side of \eqref{e:sys} under \(u\opt\) satisfies the Carath\'eodory conditions \cite[Chapter 1]{ref:Fil-88} that guarantee existence of a Carath\'eodory solution.
	\end{remark}

	\begin{remark}
		Numerical solutions to differential equations such as the ones in \eqref{e:exact:solution} are typically carried out by what are known as the shooting and multiple shooting methods. This is an active area of research; see \cite[Chapter 3]{ref:Bet-09} for a detailed discussion.
	\end{remark}

	\begin{remark}
	\label{r:extremal lift}
		The quadruple \(\bigl(\eta, p(\cdot), z\opt(\cdot), u\opt(\cdot)\bigr)\) is known as the \emph{extremal lift} of the optimal state-action trajectory \(\bigl(z\opt(\cdot), u\opt(\cdot)\bigr)\). The scalar \(\eta\) is known as the \emph{abnormal multiplier}. If \(\eta = 1\), then the extremal \(t\mapsto \bigl(\eta, p(t), z\opt(t), u\opt(t)\bigr)\) is said to be normal; if \(\eta = 0\), then the extremal is said to be abnormal. The scalar \(\eta\) is a Lagrange multiplier associated to the instantaneous cost. Interestingly, the curves for which \(\eta = 0\) are not detected by the standard calculus of variations approach \cite{ref:Cla-13}. The reason is that in calculus of variations the underlying assumption is that there are curves ``close'' to the optimal ones satisfying the same boundary conditions. But this assumption fails whenever the optimal curves are isolated in the sense that there is only one curve satisfying the given boundary conditions. In that case, a comparison between the costs corresponding to this optimal curve and other neighbouring curves turns out to be impossible to perform. The Pontryagin maximum principle, however, detects such abnormal curves and characterizes them \cite{ref:Cla-13}. At the level of generality of Proposition \ref{p:exact solution} we cannot rule out the presence of abnormal extremals in our setting.
	\end{remark}

		 Proposition \ref{p:exact solution} characterizes the necessary conditions for optimality of maps $[\tinit, \tfin]\ni t\mapsto \bigl(z(t), u(t)\bigr)$  when the map $\phi$ in \eqref{e:sys} is non-linear. In the following section, we will further examine the special case of linear plant dynamics.

\section{Linear Plant Models}
\label{s:example}
	In this section we apply the results of \secref{s:exact solution}
        to time-invariant linear systems described by: 
	\begin{equation}
	\label{e:linsys}
		\dot z(t) = \phi\big( z(t),u(t)\big)=F z(t) + G u(t),
	\end{equation}
	where \(F\in\R^{d\times d}\) and \(G\in\R^{d\times m}\) are given. As before, we assume that the time difference \(\tfin - \tinit\) is larger than the minimum duration required to execute the transfer of the state \(z(\tinit) = A\) to \(z(\tfin) = B\). Then we can use Proposition~\ref{p:exact solution} to obtain the following necessary condition for optimality:

	\begin{corollary}
	\label{c:linear case}
		Consider the optimal control problem \eqref{e:opt problem} with \(\phi\) of the form \eqref{e:linsys}. Then associated to every solution \([\tinit, \tfin]\ni t\mapsto \bigl(z\opt(t), u\opt(t)\bigr)\) to \eqref{e:opt problem} there exists a number \(\eta = 0\) or \(1\) and a vector \(\hat p\in\R^d\) such that: If \(\eta = 1\), then
			\begin{equation*}
			\begin{aligned}
				\dot z\opt(t)	& = F z\opt(t) + G u\opt(t), \qquad z\opt(\tinit) = A, \; z\opt(\tfin) = B,\\
				u\opt(t)		& \in \begin{cases}
					\displaystyle{\argmax_{v\in\admact}} \inprod{G\transp \epower{(\tfin - t)F\transp}\hat p}{v}	& \text{if }\displaystyle{\max_{v\in \admact}}\inprod{G\transp \epower{(\tfin - t)F\transp}\hat p}{v} > 1,\\
					0		& \text{otherwise}.
				\end{cases}
			\end{aligned}
			\end{equation*}
		If \(\eta = 0\), then in the above we simply have
		\[
			u\opt(t) \in \argmax_{v\in\admact} \inprod{G\transp \epower{(\tfin - t)F\transp} \hat p}{v}
		\]
		and \(\hat p \neq 0\).
	\end{corollary}

	Observe that in the normal case of \(\eta = 1\), we have sparse controls since the optimal controls are explicitly set to \(0\). We provide a proof of Corollary \ref{c:linear case} in Appendix \ref{s:app}, and note that the message of Remark \ref{r:solutions} applies accordingly to Corollary \ref{c:linear case}.

	\begin{remark}
		For the particular case where the control inputs are constrained to lie in the closed unit ball (with respect to the Euclidean norm) centered at \(0\in\R^m\), we have the particularly simple formula for the optimal control in the context of Corollary \ref{c:linear case} if \(\eta = 1\):
			\[
				u\opt(t) = 
					\begin{cases}
						0												& \text{if }\norm{G\transp \epower{(\tfin - t)F\transp}\hat p} < 1,\\
						\frac{G\transp\epower{(\tfin - t)F\transp}\hat p}{\norm{G\transp\epower{(\tfin - t)F\transp}\hat p}}	& \text{otherwise}.
					\end{cases}
			\]
			In the further special case of the control dimension being \(1\) and \(\admact = [-1, 1]\), we have
                        \begin{equation}
                          \label{eq:5}
				u\opt(t) = 
					\begin{cases}
						1	& \text{if }G\transp \epower{(\tfin - t)F\transp}\hat p \ge 1,\\
						0	& \text{if }\abs{G\transp\epower{(\tfin - t)F\transp}\hat p} < 1,\\
						-1	& \text{if }G\transp \epower{(\tfin - t)F\transp}\hat p \le -1.
					\end{cases}
   		\end{equation}
		Both the optimal controls above illustrate the \emph{bang-off-bang} nature of the optimal control mentioned in \cite[Section IV.B]{ref:NagQueNes-14}. Of course, the precise combination of the zeros and ones will depend on the initial and final states, as illustrated below.
	\end{remark}

\section{Examples}
\label{s:examples}
	We illustrate our results in this section with two examples:

	\begin{example}
	\label{example:normal}
		Consider the following scalar linear plant
		\[
			\dot z(t) = u(t),
		\]
		with initial and final conditions given by \(z(0) = 3\), \(z(5) = 0\). The set of admissible controls is given by \(\admact = [-1, 1]\). We seek a control that is feasible given the preceding conditions, and that is set to \(0\) for the maximal duration of time. In the context of this simple example it is clear that any control that is equal to \(-1\) on a Lebesgue measurable subset of \([0, 5]\) of measure \(3\) and \(0\) elsewhere is feasible. In addition, any such control achieves the minimum cost in the problem \eqref{e:Clarke opt problem}, and the corresponding minimum cost is precisely \(-2\).

		We verify the conditions of Corollary \ref{c:linear case} in the above setting: The adjoint equation is a constant since the Hamiltonian is independent of the space variable. Therefore, \(p(t) = p_0\) for some \(p_0\in\R\) and all \(t\in[0, 5]\). Since 
		\[
			u\opt(t) \in \argmax_{v\in[-1, 1]} \bigl\{ p_0 v + \eta \indic{\{0\}}(v)\bigr\},
		\]
		we have
		\[
			u\opt(t) \in \begin{cases}
				\{\sgn(p_0)\}	& \text{if }\eta = 0,\\
				\begin{cases}
					\{0\}			& \text{if }\abs{p_0} < 1,\\
					\{\sgn(p_0)\}	& \text{if }\abs{p_0} > 1,\\
					\{0, 1\}		& \text{if }p_0 = 1,\\
					\{0, -1\}		& \text{if }p_0 = -1,
				\end{cases}		& \text{if }\eta = 1.
			\end{cases}
		\]
		The first case of \(\eta = 0\) is ruled out because the corresponding constant control, regardless of the value of the constant, is not feasible. In other words, our probelm conforms to the normal case. We rule out the two constant controls corresponding to \(\abs{p_0} < 1\) and \(\abs{p_0} > 1\) since they too are not feasible. For the same reason we also eliminate all controls taking values in \(\{0, 1\}\). The only remaining possibility corresponds to any feasible control taking values in \(\{0,-1\}\). We described an uncountable family of such controls above, and therefore, each of these controls satisfies the assertions of Corollary \ref{c:linear case}.
	\end{example}

	\begin{example}
	\label{example:singular}
		Consider the following linear plant:
		\begin{equation}
		\dot{z}(t) = \begin{pmatrix}0&1\\0&0\end{pmatrix}z(t) + \begin{pmatrix}0\\1\end{pmatrix}u(t),
		\quad   z(0) = \begin{pmatrix}\xi_1\\\xi_2\end{pmatrix}.
		\label{eq:linear_plant}
		\end{equation}
		We seek a control that drives the states $t\mapsto z(t)$ from a given initial state \(z(0)\) to $z(T)=0$. The admissible action set is \(\admact = [-1, 1]\), and $T>0$ is larger than the minimum time required to enable the above manoeuvre. The control is required to be such that it is equal to \(0\) for the maximal possible duration of time. In particular, we consider the following choices
\begin{equation}
  \label{eq:4}
  T=5, \quad \xi_1=10,\, \xi_2=  -3.
\end{equation}
In the above we have
\[
F= \begin{pmatrix}0&1\\0&0\end{pmatrix} \quad \text{and}\quad G = \begin{pmatrix}0\\1\end{pmatrix}.
\]
It is immediate that \(F^2 = 0 \); hence 
                \begin{equation}
                 \epower{(\tfin - t)F\transp} = \begin{pmatrix} 1 & 0 \\ b - t & 1 \end{pmatrix}.
                \end{equation}
               In view of \eqref{e:exact:solution}, we note that the adjoint trajectory satisfies
                \begin{equation}
                \dot p(t)	 = - \begin{pmatrix} 0 & 0 \\ 1 & 0 \end{pmatrix} p(t),
                \end{equation}
                and hence \( p_1(t) = \hat p_1 \) and \(p_2(t) =  \hat p_1 ( b - t)  + \hat p_2  \) for some \(\hat p_1, \hat p_2\in\R\).

		We provide a feasible control first: Consider a control of the form
		\begin{equation}
		\label{eq:5primeprimeprime}
			u_\circ(t) = 
				\begin{cases}
					0	& \text{if } t \in [0, \theta_1[,\\
					1 	& \text{if } t \in [\theta_1, \theta_2[,\\
					0	& \text{if } t \in [\theta_2, T],\\
				\end{cases}
		\end{equation}
		for some $0 \leq \theta_1 \leq \theta_2 \leq T$ to be determined. Under this control we compute the state trajectory, and from the boundary conditions we can obtain the precise values of \(\theta_1\) and \(\theta_2\). In fact, it holds that
		\[
			z_2(t) = -3 + \int_0^t u_\circ(s)\, \dd s
		\]
		and 
		\[
		z_1(t) = 10  - 3 t + \int_0^t \int_0^s u_\circ(\tau)\,\dd\tau\,\dd s.
		\]
		Straightforward computations now lead to \( \theta_1 = \tfrac{11}{6} \) and \( \theta_2 = \tfrac{29}{6} \). The cost incurred by the above control is, therefore, \(-2\).

		We next establish that the minimum cost for our problem is precisely \(-3\).\footnote{This slick argument was pointed out to us by Witold Respondek.} Indeed, consider the evolution of the second state: \(\dot z_2(t) = u(t)\). Since \(z_2(0) = -3\), \(z_2(5) = 0\), and the admissible control set is \([-1, 1]\), it follows that any control that achieves this manoeuvre must spend at least \(3\) units of time with non-zero control values. In other words, the minimum cost is, indeed, \(-2\).

        
        Suppose \( \eta = 0 \). Then the optimal control satisfies
        \[
        	u\opt ( t) \in \argmax_{v\in[-1, 1]} \{  v(p_1 (5-t) + \hat p_2) \} 
        \]
        according to Corollary \ref{c:linear case}. Note that both \( \hat p_1, \hat p_2 \) cannot be zero simultaneously. 
		Thus \(\eta = 1\), i.e., our problem corresponds to the normal case.
                
		Using the result of Corollary \ref{c:linear case} in \eqref{eq:5} we obtain the following necessary conditions for \(\Lp 0 \)-optimal controls in this normal case:
		\begin{equation}
		\label{eq:5prime}
				u\opt(t) \in  
					\begin{cases}
						\{ 1 \}	& \text{if } \hat p_1 (5-t) + \hat p_2  > 1,\\
						\{ 0, 1 \}	& \text{if } \hat p_1 (5-t) + \hat p_2  = 1,\\
						\{ 0 \}	& \text{if }\abs{\hat p_1 (5-t) + \hat p_2} < 1,\\
						\{ -1, 0 \}	& \text{if } \hat p_1 (5-t) + \hat p_2  = -1,\\
						\{ -1 \}	& \text{if }\hat p_1 (5-t) + \hat p_2 \le -1,
					\end{cases}
   		\end{equation}
		for \(t \in [0, 5] \) and for some \( \hat p \Let (\hat p_1 , \hat p_2 ) \). In view of \eqref{e:exact:solution 2}, it is possible that both \( \hat p_1 , \hat p_2 \)  are zero, but in this case \(u\opt (t) \equiv 0 \), which is not a feasible control. Therefore, \( \hat p \neq 0 \).  Since the function \(t\mapsto \hat p_1 (5-t) + \hat p_2 \) is affine, it is monotone --- decreasing, increasing, or constant, except possibly at the instants \(t\) at which \(\hat p_1 (5-t) + \hat p_2 = 1 \) or \(-1\). Thus, in this exceptional situation, the control will be monotone almost everywhere.   Straightforward calculations exhausting all corresponding combinations of switching controls show that no such control is feasible! However if \(\hat p_1 = 0\) and \(\hat p_2\) is either \(1\) or \(-1\), we have
		\[
		u\opt (t) \in \begin{cases}
					\{0, 1 \} & \text{if } \hat p_2 = 1,\\
					\{-1, 0 \} & \text{if } \hat p_2 = -1.
					\end{cases}
		\] 
		The feasible control \(u_\circ(\cdot) \) in \eqref{eq:5primeprimeprime} satisfies this situation, and hence we conclude that \(\hat p_1 = 0 \) and \(\hat p_2 = 1 \).

        The \(\Lp 0\)-optimal control corresponding to this problem is illustrated via a solid line in Figure \ref{fig:singular}. 

		Interestingly, it follows from \cite[Control Law 8-3]{ref:AthansFalb-66} that the associated $\Lp 1$ control problem will be singular\footnote{Here ``singularity'' is meant in the sense of \cite{ref:AthansFalb-66}; it is not a universally accepted terminology!} if the components of the initial state satisfy:
		\[
		 \xi_1>\frac{\xi_2^2}{2}, \quad \xi_2<0, \quad -\frac{\xi_2}{2}-\frac{\xi_1}{\xi_2}\geq T.
		\]
		in which case the $\Lp 1$-optimal control is not necessarily $\Lp 0$-optimal. In fact, if we choose
	parameters as in~(\ref{eq:4}), then
		the $\Lp 1$-optimal control (obtained via numerical
                optimization) is as shown in dashed lines in Figure \ref{fig:singular}.
		\begin{Figure}
		\centering
		 \includegraphics[width=\linewidth]{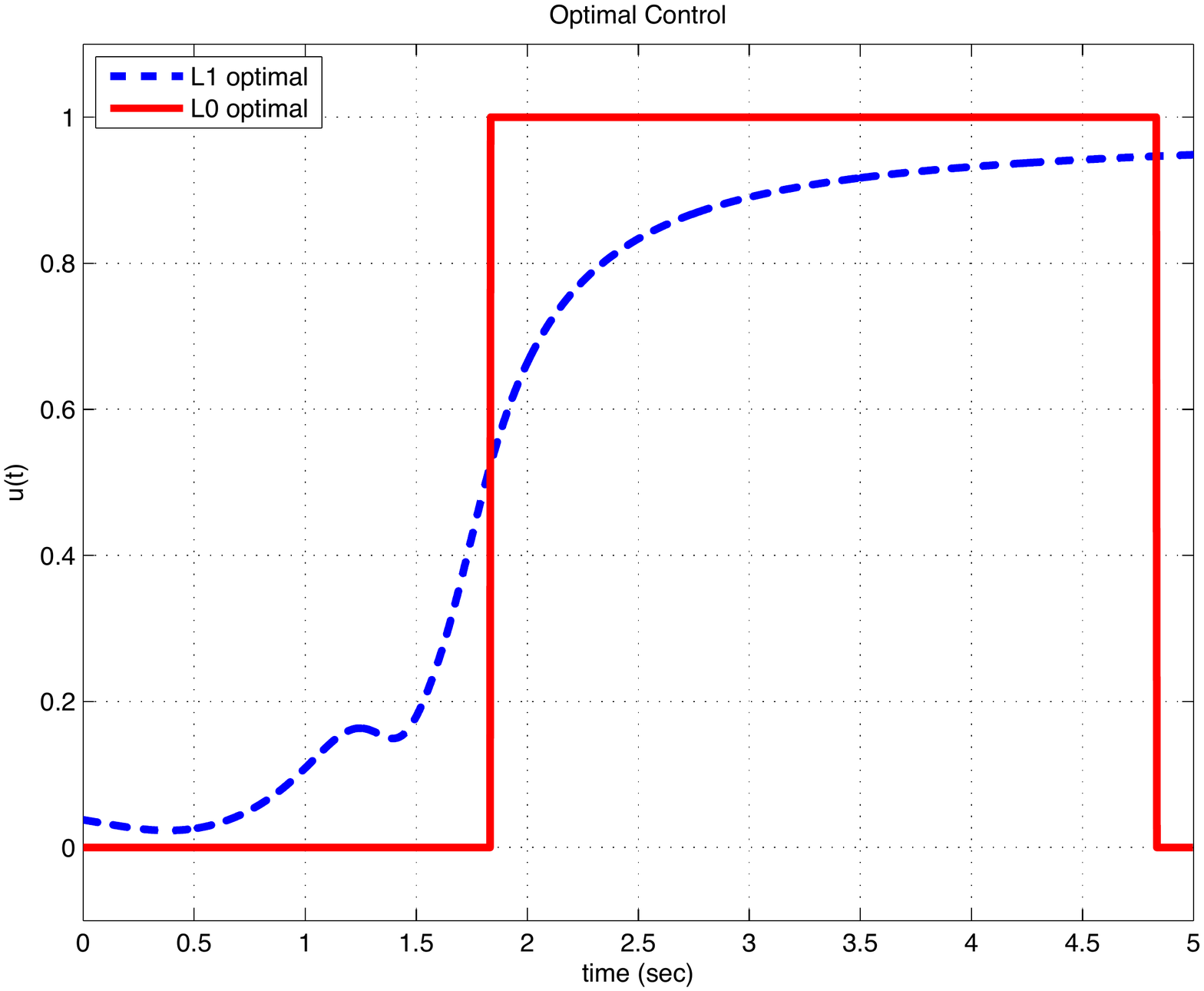}
		 \captionof{figure}{$\Lp 0$ and \(\Lp 1\)-optimal controls corresponding to Example \ref{example:singular}.}
		 \label{fig:singular}
		\end{Figure}
		Quite evidently and contrary to what one might expect (cf.\ \cite[Chapter 8]{ref:AthansFalb-66}), the \(\Lp 1\)-optimal controls are not sparse and hence cannot be $\Lp 0$-optimal. This observation opens the question of \emph{inverse optimality} of bang-off-bang solutions. It also relates to the more general problem of elucidating how the choice of cost functions leads to properties of solutions.\hfs
	\end{example}

\section{Sufficient conditions for optimality}
\label{s:suff}
	Using a non-smooth maximum principle in \secref{s:exact solution} we established the necessary conditions \eqref{e:exact:solution} for solutions to \eqref{e:opt problem}. The Pontryagin maximum principle provides necessary conditions for optimality, and as such, state-action trajectories satisfying these conditions are not necessarily optimal. In this section we provide conditions under which, in the case of our maximum hands-off control problem \eqref{e:opt problem}, the necessary conditions of the maximum principle are also sufficient for optimality.

	\begin{remark}
		This section does \emph{not} deal with existence of optimal controls; the latter appears to be a difficult problem in our case. In particular, the standard existence theorems for Bolza problems, e.g., \cite[Theorem 5.2.1]{ref:BrePic-07}, etc, do not apply directly to \eqref{e:opt problem} on account of the discontinuous nature of the instantaneous cost function \(-\indic{\{0\}}(\cdot)\).
	\end{remark}

	We have the following result:
	\begin{theorem}
	\label{t:existence}
		Consider the optimal control problem \eqref{e:opt problem} along with its associated data. Suppose that for each \(u\in\admact\) the map \(\phi(\cdot, u)\) is affine. Suppose that a normal (\(\eta = 1\)) state-action trajectory
		\[
			[\tinit, \tfin]\ni t\mapsto \bigl(z\opt(t), u\opt(t)\bigr)\in\R^d\times\admact
		\]
		satisfying the conditions of Proposition \ref{p:exact solution} exists. Then this state-action trajectory is locally optimal.
	\end{theorem}
	\begin{proof}
		By assumption \(\eta = 1\), i.e., we have the normal case, and the map \(\phi(\cdot, u)\) is affine for every \(u\in\admact\). This leads to concavity\footnote{Recall that a map \(\psi:\R^d\lra\R\) is concave if for every \(x_1, x_2\in\R^d\) and every \(\alpha\in[0, 1]\) we have \(\psi\bigl((1-\alpha)x_1 + \alpha x_2\bigr) \ge (1-\alpha)\psi(x_1) + \alpha\psi(x_2)\).} of the Hamiltonian function 
		\[
			z\mapsto H\bigl(z, p(t), u\opt(t)\bigr) = \inprod{p(t)}{\phi\bigl(z, u\opt(t)\bigr)} + \indic{\{0\}}(u\opt(t))\in\R.
		\]
		Now \cite[Theorem 24.1, Corollary 24.2]{ref:Cla-13} asserts that the state-action trajectory \([\tinit, \tfin]\ni t\mapsto \bigl(z\opt(t), u\opt(t)\bigr)\) satisfying the conditions of Proposition \ref{p:exact solution} attains a (local) minimum in \eqref{e:opt problem}.
	\end{proof}

	The difference between the above result and that in \secref{s:exact solution} lies in that Proposition \ref{p:exact solution} establishes a necessary condition, whereas Theorem \ref{t:existence} gives conditions for an optimal solution that satisfies in turn the conditions of Proposition \ref{p:exact solution}.

    We finalize our analysis by noting that the assumptions in Theorem \ref{t:existence} will be satisfied, e.g., when the problem data in \eqref{e:opt problem} is affine in the state variable.

\section{Conclusions}
\label{sec:conclusions}
	The present article has derived the exact set of necessary conditions for a control function to solve a maximum hands-off optimal control problem. The question of optimality of solutions to such problems was addressed thereafter. \(\Lp 0\)-cost optimal control problems are, of course, not limited to the class of \emph{exact} control problems that involve execution of manoeuvres under given boundary conditions in a given time. Indeed, the primary engine behind our results---the nonsmooth maximum principle---admits more general boundary conditions than the ones that we have dealt with here. Future work may include examining the question of inverse optimality of bang-off-bang controls and also investigating how the choice of instantaneous cost function influences the shape of the optimal control function.

\appendix
\section{Proofs of Proposition \ref{p:exact solution} and Corollary \ref{c:linear case}}
\label{s:app}
	We apply the non-smooth Pontryagin maximum principle \cite[Theorem 22.26]{ref:Cla-13} to the optimal control problem \eqref{e:Clarke opt problem} to characterize its solutions \([\tinit, \tfin]\ni t\mapsto \bigl(z\opt(t), u\opt(t)\bigr)\). For the sake of completeness, we adapt the non-smooth Pontryagin maximum principle  from the monograph \cite{ref:Cla-13}, to which we refer the reader for complete details including the notations.\footnote{Mention must be made of the fact that the hypotheses of \cite[Theorem 22.26]{ref:Cla-13} are considerably weaker than the hypotheses of Theorem \ref{t:Clarke extended theorem} below (which is why it is an adaptation of \cite[Theorem 22.26]{ref:Cla-13}); for the present purpose, further generality is not needed.}

	\begin{theorem}[{\cite[Theorem 22.26]{ref:Cla-13}}]
	\label{t:Clarke extended theorem}
		Consider the optimal control problem
		\begin{equation}
		\label{e:Clarke extended opt problem}
		\begin{aligned}
			\minimize_{u}	&\quad  J(x, u)  =  \int_{\tinit}^{\tfin} \Lambda\bigl(u(t)\bigr) \, \dd t \\
			\sbjto		& \quad  \begin{cases}
					\dot x(t) = f\bigl(x(t), u(t)\bigr)\quad \text{for a.e.\ }t\in[\tinit, \tfin],\\
					u:[\tinit, \tfin]\lra\admact \text{ Lebesgue measurable},\\
					\bigl(x(\tinit), x(\tfin)\bigr) \in E\subset\R^d\times\R^d,
				\end{cases}
		\end{aligned}
		\end{equation}
		where \(\Lambda:\admact\lra\R\) is bounded and lower semicontinuous,\footnote{Recall that a function \(g:\R^\nu\lra\R\) is lower semicontinuous if for every \(c \in\R\) the set \(\{y\in\R^\nu\mid g(y) \le c\}\) is closed. A function \(g:\R^\nu\lra\R\) is said to be upper semicontinuous if \(-g\) is lower semicontinuous.} \(f:\R^d\times\admact\lra\R^d\) is continuously differentiable, \(\admact\subset\R^m\) compact, and \(E\) is closed. Let \( [\tinit, \tfin]\ni t\mapsto \bigl(x\opt(t), u\opt(t)\bigr) \) be a local minimizer of \eqref{e:Clarke extended opt problem}.  For a real number \( \eta \), let the \emph{Hamiltonian} $H^\eta$ be defined by 
		\[
			H^\eta(x, p, u) = \langle p, f( x, u) \rangle - \eta \Lambda(u).
		\]
	Then there exists an absolutely continuous map \( p: [\tinit, \tfin] \to \R^n \) together with a scalar \( \eta \) equal to \(0\) or \(1\) satisfying the \emph{nontriviality condition} for all \(t\in[\tinit, \tfin]\):
		\begin{equation}
		\label{e:Clarke:nontriviality}
			\bigl(\eta, p(t)\bigr) \neq 0,
		\end{equation}
		the \emph{transversality condition}:
		\begin{equation}
		\label{e:Clarke:transversality}
			\bigl(p(\tinit), - p(\tfin)\bigr) \in  N_E^L \bigl(x\opt(\tinit), x\opt(\tfin)\bigr),
		\end{equation}
		where \(N_E^L\bigl(x\opt(\tinit), x\opt(\tfin)\bigr)\) is the limiting normal cone to \(E\) at the point \(\bigl(x\opt(\tinit), x\opt(\tfin)\bigr)\), the \emph{adjoint equation} for a.e.\ \(t\in[\tinit, \tfin]\):
		\begin{equation}
		\label{e:Clarke:adjoint}
			- \dot p(t) = \partial_x H^\eta \bigl( \cdot, p(t), u\opt(t) \bigr) (x\opt(t)),
		\end{equation}
		the \emph{Hamiltonian maximum condition} for a.e.\ \(t\in[\tinit, \tfin]\):
		\begin{equation}
		\label{e:Clarke:H max}
			H^\eta\bigl(x\opt(t), p(t), u\opt(t)\bigr) = \sup_{v\in \admact} H^\eta\bigl(x\opt(t), p(t), v\bigr),
		\end{equation} 
		as well as  the \emph{constancy of the Hamiltonian} for a.e.\ \(t\in[\tinit, \tfin]\):
		\begin{equation*}
			H^\eta\bigl(x\opt(t), p(t), u\opt(t)\bigr) = \sup_{v\in \admact} H^\eta\bigl(x\opt(t), p(t), u\bigr) = h.
		\end{equation*}
	\end{theorem}
The above non-smooth maximum principle can be used to derive the exact set of necessary conditions for maximum hands-off control \eqref{e:opt problem} as follows:

	\begin{proof}[Proof of Proposition \ref{p:exact solution}]
		We apply the non-smooth Pontryagin maximum principle Theorem \ref{t:Clarke extended theorem} to the optimal control problem \eqref{e:Clarke opt problem}. 
		For \(\eta \ge 0\) we define the \emph{Hamiltonian function} (cf.\ \cite[p.\ 464]{ref:Cla-13}) 
				\begin{multline*}
					\R^d\times\R^d\times\admact\ni(\xi, \pi, \mu) \mapsto\\
					\Hamiltonian^\eta(\xi, \pi, \mu) \Let \inprod{\pi}{\phi(\xi, \mu)} + \eta \indic{\{0\}}(\mu)\in\R.
				\end{multline*}
		In order to  derive the \emph{adjoint state equation}, we notice that for fixed \(\pi, \mu\), the function \(\R^d\ni \xi\mapsto \Hamiltonian^\eta(\xi, \pi, \mu)\) is smooth. 
				It follows that the adjoint state differential equation \eqref{e:Clarke:adjoint},\footnote{If the dynamics in \eqref{e:sys} were not smooth, then one would have a differential inclusion instead of the differential equation  \eqref{exact:adjoint}.} is given by
                \begin{equation}
                \label{exact:adjoint}
				\begin{aligned}
					\dot p(t)	& = -\partial_\xi \Hamiltonian^\eta\bigl(z\opt(t), p(t), u\opt(t)\bigr)\\
						& = - \Bigl( \partial_\xi \phi\bigl( z\opt(t), u\opt(t) \bigr) \Bigr)\transp p(t),
				\end{aligned}
				\quad \text{for a.e.\ }t\in[\tinit, \tfin].
                              \end{equation}
                              This o.d.e.\ is linear in \(p\), and due to continuous differentiability of \(\phi\), admits a unique solution on \([\tinit, \tfin]\).
            
			\par   With
                        \(E \Let \{ (A, B) \} \subset \R^d\times\R^d\) being the
                        end-points, the transversality condition \eqref{e:Clarke:transversality} to
                        \eqref{exact:adjoint} is given by
				\[
					\bigl(p(\tinit), -p(\tfin)\bigr) \in N^L_E\bigl(z\opt(\tinit), z\opt(\tfin)\bigr) = N^L_E(A, B),
				\]
				where \(N^L_E(A, B)\) is the \textsl{limiting normal cone} to \(E\) at \((A, B)\) as defined in \cite[p.\ 244]{ref:Cla-13}. Since \(E\) is a singleton, it follows from the definitions in \cite[p.\ 244, p.\ 240]{ref:Cla-13} that \(N^L_E(A, B) = \R^d\times\R^d\). In other words, the boundary conditions of the adjoint state equation \eqref{exact:adjoint} are unconstrained.
            
			The \emph{Hamiltonian maximization condition} \eqref{e:Clarke:H max} is given by
				\begin{align*}
					\Hamiltonian^{\eta}\bigl(z\opt(t),
                                        p(t), u\opt(t)\bigr)	& = \inprod{p(t)}{\phi\bigl(z\opt(t), u\opt(t)\bigr)} + \eta\indic{\{0\}}(u\opt(t))\\
						& = \sup_{v\in\admact} \Bigl\{\inprod{p(t)}{\phi\bigl(z\opt(t), v\bigr)} + \eta \indic{\{0\}}(v)\Bigr\}
				\end{align*}
				for a.e.\ \(t\in[\tinit, \tfin]\). Since the function 
				\[
					\admact\ni v\mapsto \inprod{p(t)}{\phi\bigl(z\opt(t), v\bigr)} + \eta \indic{\{0\}}(v)\in\R
				\]
				is upper semicontinuous, the supremum is attained in \(\admact\) by Weierstrass' theorem. In other words, the optimal control \(u\opt\) is given by, for a.e.\ \(t\in[\tinit, \tfin]\),
				\[
					u\opt(t) \in \argmax_{v\in\admact} \Bigl\{\inprod{p(t)}{\phi\bigl(z\opt(t), v\bigr)} + \eta \indic{\{0\}}(v)\Bigr\}.
				\]
            
                Finally, the \emph{nontriviality condition} \eqref{e:Clarke:nontriviality} states that \((\eta, p(t)) \neq (0, 0)\in\R\times\R^d\) for every \(t\in[\tinit, \tfin]\). Thus, solutions \([\tinit, \tfin]\ni t\mapsto \bigl(z\opt(t), u\opt(t)\bigr)\) to \eqref{e:Clarke opt problem} must satisfy equations \eqref{e:exact:solution}-\eqref{e:exact:solution 2}.
		\end{proof}

	\begin{proof}[Proof of Corollary \ref{c:linear case}]
		If \(\eta = 1\), for the linear case (i.e., \(\phi\) of the form \eqref{e:linsys}), the adjoint state equation is given by
		\[
			\dot p(t) = - F\transp p(t)\quad\text{for a.e.\ }t\in[\tinit, \tfin],
		\]
		which leads to the general solution 
		\[
			p(t) = \epower{-(t - \tinit)F\transp} p(\tinit)\quad\text{for all }t\in[\tinit, \tfin].
		\]
		The transversality condition \(\hat p\in\R^d\) gives the terminal condition \(p(\tfin) = \hat p\). This condition does not provide any further information about the end-point conditions for the adjoint equation. However, from the adjoint state condition  and the transversality condition we have \(\hat p = p(\tfin) = \epower{-(\tfin - \tinit)F\transp} p(\tinit)\), which shows that $$p(\tinit) = \epower{(\tfin - \tinit) F\transp} \hat p.$$ In terms of the final condition \(\hat p\),  the solution to the adjoint state equation thus reduces to
		\[
			p(t) = \epower{(\tfin - t)F\transp} \hat p\quad\text{for all }t\in[\tinit, \tfin].
		\]

		\par In view of the above, the Hamiltonian maximization condition becomes\footnote{Of course, with the normalization \(\eta = 1\), and under which we omit the superscript \(1\) on \(\Hamiltonian\).}
			\begin{align*}
				\Hamiltonian\bigl(z\opt(t), p(t), u\opt(t)\bigr)	& = \inprod{p(t)}{F z\opt(t) + G u\opt(t)} + \indic{\{0\}}(u\opt(t))\\
					& = \sup_{v\in\admact} \Bigl\{ \inprod{p(t)}{F z\opt(t) + G v} + \indic{\{0\}}(v)\Bigr\}
			\end{align*}
			for a.e.\ \(t\in[\tinit, \tfin]\). Since \(\{0\}\) is a
                        closed subset of \(\admact\), the map \(\admact\ni
                        v\mapsto \indic{\{0\}}(v)\in\R\) is an upper
                        semicontinuous function. Due to upper semicontinuity of
                        \(\admact\ni v\mapsto \inprod{p(t)}{G v} +
                        \indic{\{0\}}(v)\) and compactness of \(\admact\) (and  in view of Weierstrass' theorem), the supremum above is attained at some point of \(\admact\) for a.e.\ \(t\in[\tinit, \tfin]\).

		We conclude that the optimal control  is given by 
		\[
			u\opt(t) \in \argmax_{v\in\admact} \Bigl\{ \inprod{G\transp \epower{(\tfin - t)F\transp}\hat p}{v} + \indic{\{0\}}(v)\Bigr\}\quad \text{for all }t\in[\tinit, \tfin],
		\]
		which establishes the result.

		The case of \(\eta = 0\) is similar. The only additional observation here is that the point \(\hat p\) cannot be \(0\) for otherwise the nontriviality conditon \(\bigl(\eta, p(t)\bigr) \neq (0, 0)\in\R\times\R^d\) for all \(t\in[\tinit,\tfin]\) would be violated.
	\end{proof}

\section{Acknowledgments}
	D.\ Chatterjee was supported in part by the grant 12IRCC005SG from IRCC, IIT Bombay, India. M. Nagahara was supported in part by JSPS KAKENHI Grant Numbers 26120521, 15K14006, and 15H02668.

\end{multicols}

\end{document}